\documentclass[11pt,a4paper]{article}

\usepackage{quiver}
\input{tikzit.sty}

\tikzstyle{gate}=[shape=rectangle, text height=1.5ex, text depth=0.25ex, yshift=0.5mm, fill=white, draw=black, minimum height=5mm, yshift=-0.5mm, minimum width=5mm, font={\small}, tikzit category=circuit]
\tikzstyle{big gate}=[shape=rectangle, text height=1.5ex, text depth=0.25ex, yshift=0.5mm, fill=white, draw=black, minimum height=10mm, yshift=-0.5mm, minimum width=5mm, font={\small}, tikzit category=circuit]
\tikzstyle{Z dot}=[inner sep=0mm, minimum size=2mm, shape=circle, draw=black, fill={rgb,255: red,221; green,255; blue,221}, tikzit category=zx]
\tikzstyle{place}=[inner sep=0mm, minimum size=6mm, shape=circle, draw=black, fill=white, tikzit category=zx]
\tikzstyle{token}=[inner sep=0mm, minimum size=1mm, shape=circle, draw=black, fill=black, tikzit category=zx]
\tikzstyle{Z box}=[inner sep=0mm, minimum size=4mm, shape=rectangle, draw=black, fill={rgb,255: red,221; green,255; blue,221}, tikzit category=zx]
\tikzstyle{Z phase dot}=[minimum size=5mm, font={\footnotesize\boldmath}, shape=rectangle, rounded corners=2mm, inner sep=0.2mm, outer sep=-2mm, scale=0.8, tikzit shape=circle, draw=black, fill={rgb,255: red,221; green,255; blue,221}, tikzit draw=blue, tikzit category=zx]
\tikzstyle{X dot}=[Z dot, shape=circle, draw=black, fill={rgb,255: red,255; green,136; blue,136}, tikzit category=zx]
\tikzstyle{X phase dot}=[Z phase dot, tikzit shape=circle, tikzit draw=blue, fill={rgb,255: red,255; green,136; blue,136}, font={\footnotesize\boldmath}, tikzit category=zx]
\tikzstyle{hadamard}=[fill=yellow, draw=black, shape=rectangle, inner sep=0.6mm, minimum height=1.5mm, minimum width=1.5mm, tikzit category=zx]
\tikzstyle{paulibox}=[fill={rgb,255: red,221; green,221; blue,255}, draw=black, shape=rectangle, inner sep=0.6mm, minimum height=5mm, minimum width=5mm, font={\footnotesize}, text height=1.5ex, text depth=0.25ex, tikzit category=zx]
\tikzstyle{vertex}=[inner sep=0mm, minimum size=1mm, shape=circle, draw=black, fill=black, tikzit category=misc]
\tikzstyle{vertex set}=[inner sep=0mm, minimum size=1mm, shape=circle, draw=black, fill=white, font={\footnotesize\boldmath}, tikzit category=misc]
\tikzstyle{small black dot}=[fill=black, draw=black, shape=circle, inner sep=0pt, minimum width=1.2mm, tikzit category=circuit]
\tikzstyle{cnot ctrl}=[fill=black, draw=black, shape=circle, inner sep=0pt, minimum width=1.2mm, tikzit category=circuit]
\tikzstyle{cnot targ}=[fill=white, draw=white, shape=circle, tikzit category=circuit, label={center:$\oplus$}, inner sep=0pt, minimum width=2.1mm, tikzit fill={rgb,255: red,102; green,204; blue,255}, tikzit draw=black]
\tikzstyle{ket}=[fill=white, draw=black, shape=regular polygon, regular polygon sides=3, regular polygon rotate=-30, scale=0.7, inner sep=1pt, tikzit category=circuit, tikzit shape=rectangle, tikzit fill=green]
\tikzstyle{bra}=[fill=white, draw=black, shape=regular polygon, regular polygon sides=3, regular polygon rotate=30, scale=0.7, inner sep=1pt, tikzit category=circuit, tikzit shape=rectangle, tikzit fill=red]
\tikzstyle{scalar}=[shape=rectangle, text height=1.5ex, text depth=0.25ex, yshift=0.5mm, fill=white, draw=black, minimum height=5mm, yshift=-0.5mm, minimum width=5mm, font={\small}]
\tikzstyle{clabel}=[fill=white, draw=none, shape=rectangle, tikzit fill={rgb,255: red,56; green,255; blue,242}, font={\footnotesize}, inner sep=1pt, tikzit category=labels]
\tikzstyle{empty diagram}=[draw={gray!40!white}, dashed, shape=rectangle, minimum width=1cm, minimum height=1cm, tikzit category=misc]
\tikzstyle{amap}=[fill=white, draw=black, shape=NEbox, minimum height=1cm, minimum width=0.75cm, tikzit category=asymmetric, tikzit fill=yellow, tikzit shape=rectangle]
\tikzstyle{amap conj}=[fill=white, draw=black, shape=NWbox, tikzit category=asymmetric, tikzit fill=green, tikzit shape=rectangle]
\tikzstyle{amap adj}=[fill=white, draw=black, shape=SEbox, tikzit category=asymmetric, tikzit fill=red, tikzit shape=rectangle]
\tikzstyle{amap trans}=[fill=white, draw=black, shape=SWbox, tikzit category=asymmetric, tikzit fill=orange, tikzit shape=rectangle]
\tikzstyle{astate}=[fill=white, draw=black, shape=NEtriangle, tikzit category=asymmetric, tikzit shape=circle, tikzit fill=yellow]
\tikzstyle{astate conj}=[fill=white, draw=black, shape=NWtriangle, tikzit category=asymmetric, tikzit shape=circle, tikzit fill=green]
\tikzstyle{astate adj}=[fill=white, draw=black, shape=SEtriangle, tikzit category=asymmetric, tikzit shape=circle, tikzit fill=red]
\tikzstyle{astate trans}=[fill=white, draw=black, shape=SWtriangle, tikzit category=asymmetric, tikzit shape=circle, tikzit fill=orange]
\tikzstyle{box}=[fill=white, draw=black, shape=rectangle, minimum width=1.5cm, minimum height=1.2cm, tikzit category=circuit]
\tikzstyle{long box}=[fill=white, draw=black, shape=rectangle, minimum width=1.2cm, minimum height=2.2cm, tikzit category=circuit]
\tikzstyle{lw}=[shape=isosceles triangle, isosceles triangle stretches=true, fill=black, draw=black, minimum width=0.4cm, minimum height=3mm, inner sep=1pt, shape border rotate=180]
\tikzstyle{rw}=[shape=isosceles triangle, isosceles triangle stretches=true, fill=black, draw=black, minimum width=0.4cm, minimum height=3mm, inner sep=1pt]
\tikzstyle{dw}=[shape=isosceles triangle, isosceles triangle stretches=true, fill=black, draw=black, minimum width=0.4cm, minimum height=3mm, inner sep=1pt, shape border rotate=-90]
\tikzstyle{uw}=[shape=isosceles triangle, isosceles triangle stretches=true, fill=black, draw=black, minimum width=0.4cm, minimum height=3mm, inner sep=1pt, shape border rotate=90]
\tikzstyle{cause}=[fill={rgb,255: red,250; green,192; blue,203}, draw=black, shape=NEbox, minimum height=0.6cm, minimum width=0.8cm, tikzit category=asymmetric, tikzit fill=yellow, tikzit shape=rectangle]
\tikzstyle{MYmap}=[fill={rgb,255: red,250; green,192; blue,203}, draw=black, shape=NEbox, minimum height=0.6cm, minimum width=0.8cm, tikzit category=asymmetric, tikzit fill=yellow, tikzit shape=rectangle]
\tikzstyle{get}=[fill={rgb,255: red,150; green,222; blue,209}, draw=black, shape=NEbox, minimum height=0.6cm, minimum width=0.8cm, tikzit category=asymmetric, tikzit fill=yellow, tikzit shape=rectangle]
\tikzstyle{set}=[fill={rgb,255: red,255; green,255; blue,143}, draw=black, shape=NEbox, minimum height=0.6cm, minimum width=0.8cm, tikzit category=asymmetric, tikzit fill=yellow, tikzit shape=rectangle]
\tikzstyle{state}=[fill={rgb,255: red,224; green,176; blue,255}, draw=black, shape=NEbox, minimum height=0.6cm, minimum width=0.8cm, tikzit category=asymmetric, tikzit fill=yellow, tikzit shape=rectangle]

\tikzstyle{hadamard edge}=[-, dashed, dash pattern=on 2pt off 0.5pt, thick, draw={rgb,255: red,68; green,136; blue,255}]
\tikzstyle{box edge}=[-, dashed, dash pattern=on 2pt off 0.5pt, thick, draw={rgb,255: red,203; green,192; blue,225}]
\tikzstyle{brace edge}=[-, tikzit draw=blue, decorate, decoration={brace,amplitude=1mm,raise=-1mm}]
\tikzstyle{diredge}=[->]
\tikzstyle{double edge}=[-, double, shorten <=-1mm, shorten >=-1mm, double distance=2pt]
\tikzstyle{gray edge}=[-, {gray!60!white}]
\tikzstyle{pointer edge}=[->, very thick, gray]
\tikzstyle{boldedge}=[-, line width=1.6pt, shorten <=-0.17mm, shorten >=-0.17mm]
\tikzstyle{bidir edge}=[<->, very thick, draw={rgb,255: red,191; green,191; blue,191}]

\usepackage{epigraph}

\usepackage[utf8]{inputenc}
\usepackage[T1]{fontenc}
\usepackage{amsmath,amsfonts,amssymb,amsthm}
\usepackage{physics} 
\usepackage{mathtools}
\usepackage{graphicx}
\usepackage[svgnames,dvipsnames]{xcolor}
\usepackage{geometry}
\usepackage{fancyhdr}
\usepackage{titlesec}
\usepackage{hyperref}
\usepackage{cleveref}
\usepackage{enumitem}
\usepackage{booktabs} 
\usepackage{tikz}
\usepackage{quantikz} 
\usepackage{braket}
\usepackage{parskip}
\hypersetup{
    colorlinks=true,
    linkcolor=DarkBlue,
    citecolor=DarkRed,
    urlcolor=DarkBlue,
    pdfborder={0 0 0}
}

\titleformat{\section}
{\Large\bfseries\color{DarkBlue}}
{\thesection}{1em}{}

\titleformat{\subsection}
{\large\bfseries\color{DarkSlateBlue}}
{\thesubsection}{1em}{}

\theoremstyle{definition}
\newtheorem{definition}{Definition}[section]
\newtheorem{theorem}{Theorem}[section]
\newtheorem{lemma}{Lemma}[section]
\newtheorem{corollary}{Corollary}[section]
\newtheorem{proposition}{Proposition}[section]

\theoremstyle{remark}

\definecolor{quantumblue}{RGB}{0,100,200}
\definecolor{quantumgreen}{RGB}{0,150,100}

\title{\textbf{\Large Non-Commutation Chains in Pre- and Post-Selection Paradoxes}\\
       \vspace{1em}
       \normalsize This paper is part of a dissertation submitted in fulfillment of the requirements of a degree
      Master of Science in Mathematics and Foundations of Computer Science \\
       \vspace{0.5em}
       University of Oxford, Trinity Term 2024}
\author{Ouissal Moumou\\
        \small Mathematical Institute\\
        \small University of Oxford\\
        \small \texttt{ouissalmoumou2@gmail.com}}
\date{\today}

\begin{document}

\maketitle

\begin{abstract}
    Peculiar measurements can be obtained on systems that undergo both pre- and post-selection. We prove a conjecture from \cite{nick-paper} on logical Pre- and Post-Selection (PPS) paradoxes for a restricted case. We prove that all of these paradoxes admit non-commutation chains. We also relate this to the theory of causal balance, recently introduced in \cite{nick-paper}, and show how the theory blocks such paradoxes.
\end{abstract}

\tableofcontents
\newpage
\section{Introduction} \label{sec:intro}
Suppose a quantum system is pre-selected in the state $|\psi \rangle$ at time $t_0$ and post-selected in the state $\langle \phi|$ at time $t_2$. Suppose also that the system undergoes an intermediate projective measurement $\{ \mathit{P}\}$ at time $t_1$. One could calculate the probability of the intermediate measurement by conditioning both on the pre- and post-selection, this is the main spirit behind the ABL rule introduced in \cite{abl-aharonov}. 

For a degenerate operator $C$ at time $t$, the probability of obtaining an outcome $c_n$ is specified by the ABL rule as
\begin{equation}
\label{abl-rule}
    P(C) = \frac{|\langle \phi| P_{C = c_n} |\psi \rangle|^{2}}{\sum_{i}|\langle \phi|P_{C = c_i}|\psi \rangle|^{2}},
\end{equation}
where $P_{C = c_i} = \sum_{i} |\Phi \rangle \langle \Phi|$ is the projection on the states with eigenvalue $c_i$.

Peculiar situations arise for specific choices of $|\psi\rangle$, $\{ \mathit{P}\}$, and $\langle \phi|$ in which the probabilities obtained are non-intuitive. These situations are referred to as pre- and post-selection paradoxes. 

\cite{nick-paper} posited that these paradoxes are absent within the framework of causal balance theory, as this theoretical framework inherently avoids what we term ``non-commutation chains'' in the present work. A non-commutation chain, which we define rigorously in \Cref{sec:chains}, arises when the intermediate projection operator in a pre- and post-selection scenario fails to commute with both the pre-selected and post-selected quantum states. We establish a proof for a restricted subset of Ormrod's conjecture from \cite{nick-paper}, demonstrating that all paradoxes within this constrained class exhibit a systematic occurrence of non-commutation chains.

We start with the quantum pigeonhole principle as an exemplar of these paradoxes to provide intuitive understanding, followed by a formal definition of pre- and post-selection phenomena that establishes the theoretical foundation for our subsequent proof. We conclude with a discussion of the implications and directions for future research.

\subsection{The Quantum Pigeonhole Principle Paradox}
\cite{pigeon} introduced the quantum pigeonhole principle in which $3$ particles are prepared in a superposition state of being in $2$ different boxes. Boxes $1$ and $2$ are represented by $|0\rangle$ and $|1\rangle$ respectively. A superposition of being in both boxes $1$ and $1$ would then be represented more conveniently with $|+\rangle = \frac{1}{\sqrt{2}} (|0\rangle + |1\rangle)$. Therefore, the states corresponding to the pre- and post-selection states are
\begin{align*}
    | \psi \rangle &= {|+\rangle}_1 {|+\rangle}_2 {|+\rangle}_3, \\
    \langle \phi | &= {\langle i |}_1 {\langle i |}_2 {\langle i |}_3,
\end{align*}
respectively, where $\langle i| = \frac{1}{\sqrt{2}} (\langle 0 | + i\langle 1|)$.

We aim to check whether two particles are in the same box. Since the three particles are in the same initial and final states, we will only demonstrate the paradox for particles 1 and 2. By symmetry, the same applies for particles 2 and 3 and particles 1 and 3. Particles 1 and 2 being in the same box means that they are either both in box 1, or both in box 2 which we represent with projectors $P_{11} ={|0\rangle}_1 {|0\rangle}_2 {\langle 0|}_1 {\langle 0|}_2$ and $P_{22} = {|1\rangle}_1 {|1\rangle}_2 {\langle 1|}_1 {\langle 1|}_2$ respectively. On the other hand, particles 1 and 2 being in separate boxes means particle 1 is in box 1 and particle 2 is in box 2, or vice versa, which we represent with projectors $P_{12} = {|0\rangle}_1 {|1\rangle}_2 {\langle 1|}_1 {\langle 0|}_2$ and $P_{21} = {|1\rangle}_1 {|0\rangle}_2 {\langle 1|}_1 {\langle 0|}_2$ respectively.

Consequently, the two particles under consideration being in the same box corresponds to the projector $P_\text{same} = P_{11} + P_{22}$. Similarly, the two particles being in different boxes correspond to the projector $P_\text{diff} = P_{12} + P_{21}$. Because the value of $\langle \phi | P_\text{same} \ket{\psi}$ turns out to be zero, it follows that the probability of finding particles 1 and 2 in the same box using the ABL rule is also zero. Considering the symmetry in our example with the other particles, this means that the probability of finding particles 2 and 3 in the same box in the intermediate measurement is also zero, and so is the probability of finding both particles 1 and 3 in the same box. One then is prompted to conclude that it is with certainty that we observe that no two particles can be found in the same box. However, recalling the very basic yet powerful pigeonhole principle, it is quite peculiar to have two boxes and three particles. Yet, no two particles share the same box! 

\section{Pre- and Post-Selection Paradoxes: A Formal Definition} \label{sec:paradoxes}
Below is a formal definition of a pre- and post-selection paradox inspired from \cite{contexuality}. Consider a Hilbert space, a choice of pre-selection $|\psi\rangle$ and post-selection $\langle\phi|$, and let $\mathcal{P}$ be a finite set of projectors closed under complements composed of the projectors related to the pre- and post-selection system $\{P,I - P\}$ that are uniquely determined by the projectors $P$. We only consider the cases in which the ABL probabilities corresponding to using the projectors $\{P,I - P\}$ yield 0 or 1 values (hence the ``logical'' part).

    We would like to be able to generate a partial Boolean algebra $\mathcal{P'}$ from $\mathcal{P}$ using the following for any $P$ and $Q$ in $\mathcal{P'}$:
    \begin{itemize}
        \item If $P, Q \in \mathcal{P'}$ and $PQ = QP$, then $PQ \in \mathcal{P'}$,
        \item If $P \in \mathcal{P'}$, then $I - P \in \mathcal{P'}$.
    \end{itemize}

    One intuitive way of understanding the new partial Boolean algebra is by considering projectors corresponding to propositions, and the extension to the partial Boolean algebra is our way of wanting to also take disjunctions and conjunctions of the propositions at hand (we originally only had the propositions and their complements given to us by the experiments).
    Suppose we want to extend the probability function $f$ given to us by the ABL rule from $\mathcal{P}$ to $\mathcal{P'}$. In other words, we want to find a probability function on $\mathcal{P'}$ that recovers the probability function defined on $\mathcal{P}$, such that the following algebraic conditions are satisfied:
    \begin{enumerate}[label=(\roman*)]
        \item For all $P \in \mathcal{P'}$, $0 \leq f(P) \leq 1$,
        \item $f(I) = 1$ and $f(0) = 0$,
        \item For all $P, Q \in \mathcal{P'}$, $f(P + Q - PQ) = f(P) + f(Q) - f(QP)$.
    \end{enumerate} 
    
\begin{definition}
\label{def-pps}
(PPS Paradox)
Assuming the above setting, we say that the ABL predictions for $\mathcal{P}$ form a logical PPS paradox when we fail to find a function that fails one or more of the algebraic conditions above.
\end{definition}

Applying this to the case of the pigeonhole principle paradox: the projectors corresponding to two particles being in different boxes $P_{\text{diff}_{1, 2}}$, $P_{\text{diff}_{1, 2}}$, and $P_{\text{diff}_{1, 2}}$. We know that: $f(P_{\text{diff}_{1, 2}}) = 1$, $f(P_{\text{diff}_{1, 2}}) = 1$, and $f(P_{\text{diff}_{1, 2}}) = 1$.

We know that $p(e) = 1$ and $p(f) = 1 \implies p(e \land f) = 1$. This can be extended in our case of three events, and knowing that the probabilities for these events are all $1$, we can conclude that
    \begin{equation*}
        f(P_{\text{diff}_{1, 2}} P_{\text{diff}_{2, 3}} P_{\text{diff}_{1, 3}}) = 1.
    \end{equation*}

However, note that: $P_{\text{diff}_{1, 2}} P_{\text{diff}_{2, 3}} P_{\text{diff}_{1, 3}} = 0$. Therefore, $f(0) = 1$, which violates condition $(i)$ from \Cref{def-pps}. Thus, the quantum pigeonhole principle is just another instance of logical PPS paradoxes.

\section{PPS Paradoxes Admit Non-Commutation Chains} \label{sec:chains}
We start with the following definition of a non-commutation chain:
\begin{definition} (Non-Commutation Chain)
\label{non-commutation chain}
A projector $P$ has a non-commutation chain with two other projectors $Q_1$ and $Q_2$ iff $Q_1 P \neq P Q_1$ and $Q_2 P \neq P Q_2$.
\end{definition}

The following lemma will be useful later and is proved in \Cref{app:proof2}:
\begin{lemma}
\label{P-commutes}
    In the context of a logical pre- and post-selection paradox with pre- and post-selection projectors $P_\psi = |\psi \rangle \langle \psi|$ and $P_\phi = |\phi \rangle \langle \phi|$, respectively, a projector $P$ corresponding to an intermediate measurement at time $t_1 < t < t_2$ does not have a non-commutation chain with $P_\psi$ and $P_\phi$ iff
    \begin{itemize}
        \item $P |\psi \rangle = | \psi \rangle$, and we say that $P$ and $P_\psi$ idempotently commute, or
        \item $P |\psi \rangle = 0$, and we say that $P$ and $P_\psi$ orthogonally commute, or 
        \item $P |\phi \rangle = | \phi \rangle$, and we say that $P$ and $P_\phi$ idempotently commute, or
        \item $P |\phi \rangle = 0$, and we say that $P$ and $P_\phi$ orthogonally commute.
    \end{itemize}
\end{lemma}

\begin{corollary}
\label{idem-ortho}
    Two projectors $P$ and $P_\psi$ idempotently commute (in the sense of lemma \ref{P-commutes}) iff $I - P$ and $P_\psi$ orthogonally commute.
\end{corollary}

This corollary says the same thing about the relationship between $P$ and $P_\phi$ if they idempotently, or orthogonally commute. 

\begin{theorem} (Non-Commutation Chains for Logical PPS Paradoxes) 
    Consider a pre- and post-selection scenario with corresponding intermediate measurement sets $\mathcal{P}$ and $\mathcal{P'}$ per definition \ref{def-pps} and pre- and post-selection rank-1 projectors $P_{\psi} = | \psi \rangle \langle \psi |$ and $P_{\phi} = | \phi \rangle \langle \phi |$ such that
    \begin{itemize}
        \item $|\psi \rangle$ and $\langle \phi |$ are the pre- and post-selection states,
        \item All the projectors in $\mathcal{P}$ commute,
        \item $\mathcal{P'}$ is a finite set.
    \end{itemize}
    
    Every logical PPS paradox with the specifications above has at least one projector in $\mathcal{P}$ that forms a non-commutation chain with $P_\psi$ and $P_\phi$.
\end{theorem}

\begin{proof}
    We prove the contrapositive of this statement. We assume that we have a scenario with all the specifications above where every projector in $\mathcal{P}$ does not form a non-commutation chain with $P_\psi$ and $P_\phi$, and we prove that we never get a logical PPS paradox. 

    Consider the projectors $Q_j$ that can be written as the product of $n$ projectors in $\mathcal{P}$:
    \begin{equation}
        Q_j = \tilde{P_1}...\tilde{P_n},
    \end{equation}
     such that $\tilde{P_i} = P_i \in \mathcal{P}$ or $\tilde{P_i} = I - P_i \in \mathcal{P}$. In other words, take all the $n$ projectors corresponding to intermediate measurements, give each projector an index from $1$ to $n$. In each index, either put a projector or its complement, then take the product of all the projectors altogether. One can see that there are $2^n$ possibilities for this arrangement. The $Q_j$s are all orthogonal to each other i.e. for any $i$ and $j$ such that $k \neq j$: $Q_k Q_j = \delta_{kj} Q_j$ and $\sum_{j}{Q_j} = I$. Note that every $Q_j$ is a projector because it is made of the product of projectors from $\mathcal{P}$, all of which commute with each other.
    
    Let $\Omega$ be the sample space containing all the possible $Q_j$s, i.e $\Omega = \{Q_j\}_j$. As shown in \Cref{app:events}, we can take $\mathcal{E} = 2^\Omega$ as an event space. However, as also shown in \Cref{app:events}, there is an equivalent event space $\mathcal{E}$ that we obtain by summing distinct elements in $\Omega$ such that
    \begin{equation}
        \mathcal{E} = \left\{\sum_{j \in S}{Q_j} / S \in 2^\Omega\right\}.
    \end{equation}

    As an application of what is shown in \Cref{app:events}, $\text{Poweset}(\Omega)$ and $\mathcal{E}$ are not only equivalent, but $\mathcal{E}$ also have a Boolean algebra structure. The equivalence between $\text{Poweset}(\Omega)$ and $\mathcal{E}$  is defined as for two commuting projectors $P$ and $Q$ in $\mathcal{E}$:
    \begin{itemize}
        \item $PQ = \sum_{j \in S_P \cap S_Q}{Q_j}$,
        \item $P + Q - PQ = \sum_{j \in S_P \cup S_Q}{Q_j}$,
        \item $\overline{P} = \sum_{j \in \overline{S}_P}{Q_j}$.
    \end{itemize}

    Before we move to the first step in our proof, we will show that $\mathcal{P} \subseteq \mathcal{E}$. Let $P_k$ be any projector in $\mathcal{P}$, and let's prove that $P_k \in \mathcal{E}$. We can write $P_k$ as
    $$P_k = \sum_{j}{\tilde{P_1}... P_k...\tilde{P_n}}
        = \sum_{j \in S_{P_k}}{Q_j} \in \mathcal{E}.$$
    Therefore, $\mathcal{P} \subseteq \mathcal{E}$.

    As a first step in our proof, we will show that $\mathcal{P'} = \mathcal{E}$.

    For the first direction $\implies$, we show that $\mathcal{P'} \subseteq \mathcal{E}$.  We know that $\mathcal{P} \subseteq \mathcal{E}$, and we know that by Definition \ref{def-pps}, $\mathcal{P'}$ is the smallest partial Boolean algebra from $\mathcal{P}$. Knowing that $\mathcal{E}$ is a Boolean algebra that contains $\mathcal{P}$, then it has to contain the smallest one formed by $\mathcal{P}$, therefore, it contains $\mathcal{P'}$, and thus: $\mathcal{P'} \subseteq \mathcal{E}$.

    For the opposite direction $\impliedby$, we show that $\mathcal{E} \subseteq \mathcal{P'}$. Let $P$ be any projector in $\mathcal{E}$, $P$ can be one of two cases:
    \begin{enumerate}
        \item $P = Q_j$ for some $Q_j = \tilde{P_1}...\tilde{P_n}$ which means $P$ can be written as a conjunction of projectors in $\mathcal{P}$,
        \item $P = \sum_{j \in S}{Q_j}$ which means $P$ can be written as a disjunction of elements in $\Omega$, more specifically, as a disjunction of conjunction if elements of $\mathcal{P}$.
    \end{enumerate}

    In the first case, we know that any conjunction of elements in $\mathcal{P}$ is in $\mathcal{P'}$ since all the projectors in $\mathcal{P}$ commute, and by definition, $\mathcal{P'}$ contains all the products of projectors in $\mathcal{P}$ that commute. Now that we know that any $Q_j$ is in $\mathcal{P'}$, if $P$ is a conjunction of elements in $\Omega$ i.e $P = \sum_{j \in S}{Q_j}$, then it has to be in $\mathcal{P'}$ since by definition, $\mathcal{P'}$ is closed under conjunction, and so it contains any possible conjunction of two projectors in $\mathcal{P'}$ . Therefore, $\mathcal{E} \subseteq \mathcal{P'}$. Thus $\mathcal{E} = \mathcal{P'}$.e pre-selection and post-selection respectively. <---- typo here??

    We are at the main part of our proof of the theorem, the one involving defining a probability distribution on $\mathcal{P'}$. Consider the function $f$ defined on $\mathcal{P'}$ such that for any projector $P$ in $\mathcal{P'}$,
    \begin{equation}
        f(P) = \frac{\text{Tr}(P_\psi P P_\phi)}{\text{Tr}(P_\psi P_\phi)}.
    \end{equation}
    We will show that this function satisfies all the conditions for being a probability measure on $\mathcal{P'}$ as per \Cref{def-pps}.

    For condition $(ii)$ from the definition, we have that
    $$f(I) = \frac{\text{Tr}(P_\psi I P_\phi)}{\text{Tr}(P_\psi P_\phi)} = 1, \hspace{8em} f(0) = \frac{\text{Tr}(P_\psi 0 P_\phi)}{\text{Tr}(P_\psi P_\phi)} = 0.$$

    Condition $(iii)$ follows directly from the linearity of the trace function. The trace operation is linear, so
    \begin{align*}
        f(P + Q - PQ) &= \frac{\text{Tr}(P_\psi (P + Q - PQ) P_\phi)}{\text{Tr}(P_\psi P_\phi)} \\
        &= \frac{\text{Tr}((P_\psi P P_\phi) + (P_\psi Q P_\phi) - (P_\psi PQ P_\phi))}{\text{Tr}(P_\psi P_\phi)} \\
        &= f(P) + f(Q) - f(PQ).
    \end{align*}
    For condition $(i)$, show that for any projector $P$ in $\mathcal{P'}$, $0 <= f(P) <= 1$. Since $\mathcal{P'} = \mathcal{E}$, any projector $R$ in $\mathcal{P'}$ can be written as a product of projectors in $\mathcal{P}$ such that: $R = Q_j = \tilde{P_1}...\tilde{P_n}$ or as the sum of several $Q_i$ as shown above. Now,
    \begin{align*}
        f(Q_j) &= \frac{\text{Tr}(P_\psi \tilde{P_1}...\tilde{P_n} P_\phi)}{\text{Tr}(P_\psi P_\phi)} .
    \end{align*}
    Since all the elements in $\mathcal{P}$ commute, we can arrange the elements in any product $\tilde{P_1}... \tilde{P_k} \tilde{P_{k+1}}...\tilde{P_n}$ such that the first $k$ projectors are the ones that commute with $P_\psi$, and that the rest of the projectors all commute with $P_\phi$. Now, using \Cref{P-commutes}, we know that for the first $k$ projectors in $Q_i$,
    \begin{align*}
        P_\psi \tilde{P_i} = |\psi \rangle \langle \psi| P = P_\psi \tilde{P_i} = \begin{cases}
            P_\psi & \text{if } \tilde{P_i} \text{ and } P_\psi \text{ idempotently commute}, \\
            0 & \text{if } \tilde{P_i} \text{ and } P_\psi \text{ orthogonally commute}.
        \end{cases}
    \end{align*}
    One can see that applying this repeatedly for all the projectors $P_1$ to $P_k$ can only result in $P_\psi$ if all the projects tend not to be orthogonal with $P_\psi$. In the opposite case, the series might collapse to 0. 
    
    Similarly, for the rest of the projectors, we have $\tilde{P_i} P_\phi  = P_\phi$ or $\tilde{P_i} P_\phi = 0$. Now, coming back to $f(Q_j)$, it can only be 1 if all the projectors happen to commute non-orthogonally with either $P_\psi$ or $P_\phi$.

     Now, we analyze the second case in which elements of $\mathcal{E}$ can be written as sums of distinct $Q_j$. The summands in this case are distinct. In other words, each two summands are different by at least one $\tilde{P_i}$ (in fact, each two summands are orthogonal). Moreover, recalling Corollary \ref{idem-ortho} we know that in our non-commutation chain scenario, if a certain projector $Q$ commutes with $P_\psi$ commute non-orthogonally, then its complement $I - Q$ commutes orthogonally with $P_\psi$, and vice versa. The same can be said about $Q$ and $P_\phi$ if $Q$ commutes with $P_\phi$. This means that considering all possible $Q_j$ configurations (by configuration here we mean, different choices for $\tilde{P_i}$), there is only one single configuration of $\tilde{P_i}$ where all the projectors non-orthogonally commute with either $P_\psi$ or $P_\phi$, only and only in this single case would $f$ be $1$. In all the others it will be $0$. Knowing that all the summands are distinct, only one configuration equalling 1 means that any sum will always be 1 or 0. Therefore, we have just shown an even stronger claim than the one required by condition $(i)$, that for any projector $P$ in $\mathcal{P'}$, $f(P) = 0$ or $f(P) = 1$. 
\end{proof}

\section{Connection to the Theory of Causal Balance} \label{sec:causalbalance}
The theorem proved in the preceding section draws from the theory of causal balance. The proof approach can be extended to establish the conjecture proposed in \cite{nick-paper}, namely that the theory of causal balance blocks both pre-logical and post-logical paradoxes. We demonstrate this connection in the following analysis. The requisite background on causal balance theory is provided in \Cref{app:causalbalance}.

We recall Theorem 4 in \cite{nick-paper}.
\begin{theorem}
\label{allowed-interferences}
In a circuit that models a causal structure, the only interference influences allowed are of the form
$$\left\{ P_{A_m^{in}}^{e_m'} \right\} \rightarrow \left\{ P_{A_k^{out}}^{e_k'}\right\},$$
such that $m \leq k$.
\end{theorem}

Consider the following unitary circuit:
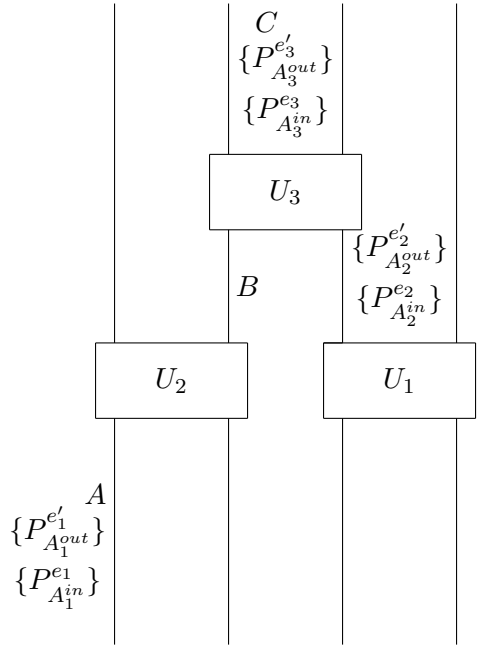
\begin{center}
    \begin{tikzpicture}
	\begin{pgfonlayer}{nodelayer}
		\node [style=none] (0) at (-15.5, 5) {};
		\node [style=none] (1) at (-13.5, 5) {};
		\node [style=none] (2) at (-13.5, 4) {};
		\node [style=none] (3) at (-15.5, 4) {};
		\node [style=none] (4) at (-12.5, 5) {};
		\node [style=none] (5) at (-10.5, 5) {};
		\node [style=none] (6) at (-10.5, 4) {};
		\node [style=none] (7) at (-12.5, 4) {};
		\node [style=none] (8) at (-14, 6.5) {};
		\node [style=none] (9) at (-14, 7.5) {};
		\node [style=none] (10) at (-12, 7.5) {};
		\node [style=none] (11) at (-12, 6.5) {};
		\node [style=none] (12) at (-13.75, 5) {};
		\node [style=none] (13) at (-13.75, 6.5) {};
		\node [style=none] (14) at (-12.25, 6.5) {};
		\node [style=none] (15) at (-12.25, 5) {};
		\node [style=none] (16) at (-10.75, 5) {};
		\node [style=none] (17) at (-12.25, 7.5) {};
		\node [style=none] (18) at (-13.75, 7.5) {};
		\node [style=none] (19) at (-15.5, 5) {};
		\node [style=none] (20) at (-15.25, 4) {};
		\node [style=none] (21) at (-13.75, 4) {};
		\node [style=none] (22) at (-15.25, 5) {};
		\node [style=none] (23) at (-15.25, 9.5) {};
		\node [style=none] (24) at (-13.75, 9.5) {};
		\node [style=none] (25) at (-12.25, 9.5) {};
		\node [style=none] (26) at (-10.75, 9.5) {};
		\node [style=none] (27) at (-15.25, 1) {};
		\node [style=none] (28) at (-13.75, 1) {};
		\node [style=none] (29) at (-12.25, 1) {};
		\node [style=none] (30) at (-10.75, 1) {};
		\node [style=none] (31) at (-10.75, 4) {};
		\node [style=none] (32) at (-12.25, 4) {};
		\node [style=none] (33) at (-13, 7) {$U_3$};
		\node [style=none] (34) at (-14.5, 4.5) {$U_2$};
		\node [style=none] (35) at (-11.5, 4.5) {$U_1$};
		\node [style=none] (36) at (-16, 1.75) {$\{P_{A_1^{in}}^{e_1}\}$};
		\node [style=none] (37) at (-16, 2.5) {$\{P_{A_1^{out}}^{e_1'}\}$};
		\node [style=none] (38) at (-11, 5.25) {};
		\node [style=none] (40) at (-11.5, 5.5) {$\{P_{A_2^{in}}^{e_2}\}$};
		\node [style=none] (41) at (-11.5, 6.25) {$\{P_{A_2^{out}}^{e_2'}\}$};
		\node [style=none] (43) at (-13, 8) {$\{P_{A_3^{in}}^{e_3}\}$};
		\node [style=none] (44) at (-13, 8.75) {$\{P_{A_3^{out}}^{e_3'}\}$};
		\node [style=none] (45) at (-15.5, 3) {$A$};
		\node [style=none] (46) at (-13.5, 5.75) {$B$};
		\node [style=none] (47) at (-13.25, 9.25) {$C$};
	\end{pgfonlayer}
	\begin{pgfonlayer}{edgelayer}
		\draw (9.center) to (10.center);
		\draw (10.center) to (11.center);
		\draw (11.center) to (8.center);
		\draw (8.center) to (9.center);
		\draw (0.center) to (1.center);
		\draw (1.center) to (2.center);
		\draw (2.center) to (3.center);
		\draw (3.center) to (0.center);
		\draw (4.center) to (5.center);
		\draw (5.center) to (6.center);
		\draw (7.center) to (4.center);
		\draw (7.center) to (6.center);
		\draw (23.center) to (22.center);
		\draw (12.center) to (13.center);
		\draw (4.center) to (15.center);
		\draw (15.center) to (14.center);
		\draw (16.center) to (26.center);
		\draw (20.center) to (27.center);
		\draw (21.center) to (28.center);
		\draw (32.center) to (29.center);
		\draw (30.center) to (31.center);
		\draw (17.center) to (25.center);
		\draw (18.center) to (24.center);
	\end{pgfonlayer}
\end{tikzpicture}

\end{center}

\Cref{allowed-interferences} dictates that there can never be an influence from $\{P_{A_2^{out}}^{e_2'}\}$ to $\{P_{A_1^{in}}^{e_1}\}$ for two reasons: 
\begin{enumerate}
    \item It is an influence from an ``out'' projector to an ``in'' projector.
    \item It is an influence from a projector that is higher up in the circuit to a projector that is lower up in the circuit. Interference influences are equivalent to commutations as shown in theorem \ref{non-commutation chain}.
\end{enumerate}

One of the patterns that are immediately eliminated by \Cref{allowed-interferences} is non-commutation chains. 

Given the assumptions in \Cref{non-commutation chain}, every PPS paradox admits a non-commutation chain, the absence of which is guaranteed by the theory of causal balance. Therefore, one could say that the theory of causal balance ``blocks'' PPS paradoxes. In other words, PPS paradoxes, never occur under the framework of the theory of causal balance.

Intriguingly, \cite{nick-paper} demonstrated that any phenomenon from standard quantum theory can be reproduced within the framework of the theory of causal balance. This raises a compelling question: given our analysis showing that the theory prevents these paradoxes (at least in the specific case examined in this paper), how might one model such paradoxes within this framework? We leave this investigation for future work.

\section{Conclusion}
This paper establishes that logical PPS paradoxes in the restricted case where all projectors in $\mathcal{P}$
 commute and where $\mathcal{P'}$ is finite necessarily admit non-commutation chains. While this result provides important structural insight into the nature of these paradoxes, several theoretical limitations warrant further investigation.

The finite cardinality constraint on $\mathcal{P'}$
 represents a significant restriction, as $\mathcal{P'}$ can in principle be infinite. Whether our theorem extends to the infinite case remains an open question, and failure to do so would reveal a fundamental relationship between the cardinality of $\mathcal{P'}$ and the structural properties of these paradoxes. Additionally, our analysis assumes commutativity of all projectors in $\mathcal{P}$, though it remains unclear whether scenarios with non-commuting projectors are physically realizable in logical PPS frameworks.

Our findings also bear on the relationship between logical PPS paradoxes and the theory of causal balance. Given that \cite{nick-paper} demonstrated the theory's ability to reproduce all standard quantum phenomena, yet our analysis shows these particular paradoxes admit non-commutation chains that should be blocked by causal constraints, an interesting question emerges regarding how such paradoxes might be modeled within this framework. We conjecture that circuit representations of PPS paradoxes (in the context of the theory of causal balance), which necessarily incorporate unitary interactions, will exhibit richer causal structures than the simple sequential projector arrangements. This suggests that modeling PPS paradoxes within the theory of causal balance may reveal fundamentally different structural properties than those apparent in the standard formulation, representing an important direction for future work.

\section*{Acknowledgments}

I would like to thank my supervisors Dr. Nick Ormrod and Professor Jonathan Barret. Endless thanks to my friends and family for their support. Thanks to the Optiver Foundation Scholarship for financially supporting my MSc at the University of Oxford.

\bibliographystyle{plain}
\bibliography{references}

\appendix
\section{Proof of Lemma \ref{P-commutes}} \label{app:proof1}
\begin{proof}
    For the $\impliedby$ direction, we show below what each of the cases implies about commuting with either $P_\psi$ or $P_\phi$.
    \begin{itemize}
        \item If $P |\psi \rangle = | \psi \rangle$, then,
        $$ P P_\psi =P | \psi \rangle \langle \psi | \implies P P_\psi = |\psi \rangle \langle \psi |.$$
        we know that $\langle \psi| = \langle \psi| P$. We deduce that  $P P_\psi = |\psi \rangle \langle \psi | P = P_\psi P$. Therefore $P$ and $P_\psi$ commute.

        \item If $P |\psi \rangle = 0$, then
        $$P P_\psi = P |\psi \rangle \langle \psi| \implies P P_\psi = 0.$$
        Now, for the other side, we have $P_\psi P = |\psi \rangle \langle \psi|P$. We know from proposition 2 that because $P |\psi \rangle = 0$, $\langle \psi| P = 0$. Therefore, $P_\psi P = 0 = P P_\psi$, and $P$ and $P_\psi$ commute.
    \end{itemize}
        We can prove the other two cases involving $P_\phi$ and $P$ with a similar argument. We can see that in all the cases, we can deduce the commutation of $P$ with either $P_\psi$ or with $P_\phi$.

    For the $\implies$ direction, we assume that $P$ and $P_\psi$ commute or $P$ and $P_\phi$ commute, and we show that either $P$ and $P_\psi$ orthogonally commute or idempotently commute, or that $P$ and $P_\phi$ orthogonally commute or idempotently commute. 
    
    \textbf{Case 1: $P$ and $P_\psi$ commute}\\
    Two projectors commuting means that they have a shared basis. Let $|e_1 \rangle, |e_2 \rangle, ..., |e_n \rangle$ be the shared basis of both projectors. 

    Now, since $P_\psi = |\psi \rangle \langle \psi |$ is a rank-1 projector, there is a $k \in \{1, ..., n\}$ such that: $P_\psi = |e_k \rangle \langle e_k|$. Considering that $P$ is also a rank-1 projector, we have: $P = |e_j \rangle \langle e_j|$ for some $j \in \{1, 2, ..., n\}$. Therefore, $P |\psi \rangle = |e_j \rangle \langle e_j| |e_k \rangle$. Now, since $|e_j\rangle$ and $|e_k \rangle$ are basis states, there are two possibilities:
    \begin{itemize}
        \item \textbf{$k = j$}, in which case $\langle e_j| |e_k \rangle = 1 \implies P |\psi\rangle = |e_j \rangle \langle e_j| |e_k \rangle = |e_j \rangle = |e_k\rangle = |\psi \rangle$,
        \item \textbf{$k \neq j$}, in which case $\langle e_j| |e_k \rangle = 0 \implies P |\psi\rangle = |e_j \rangle \langle e_j| |e_k \rangle = 0$.
    \end{itemize}

    Therefore, $P |\psi \rangle = |\psi \rangle$ or $P |\psi \rangle = 0$.
    
    \textbf{Case 2: $P$ and $P_\phi$ commute}\\
    We use a similar argument to the one used in the previous case, and we deduce that we obtain one of the following cases: $P |\phi \rangle = |\phi\rangle$ or $P |\phi \rangle = 0$.
    
    More intuitively, if $P$ and $P_\psi$ commute,  since we are dealing with rank-1 projectors, one can try to imagine the subspaces onto which both $P$ and $P_\psi$ project. Their corresponding subspaces will either have no overlap (the orthogonal case), or one projector's subspace will contain the other projector's subspace (the idempotent case).
\end{proof}

\section{Proof of Corollary \ref{idem-ortho}} \label{app:proof2}
\begin{proof}
    ($\Longrightarrow$) We assume that $P$ and $P_\psi$ idempotently commute. $P |\psi \rangle = | \psi \rangle \implies (I - P) |\psi \rangle = |\psi \rangle - P |\psi \rangle \implies (I - P) |\psi \rangle = |\psi \rangle - |\psi \rangle = 0$. Therefore $I - P$ and $P_\psi$ orthogonally commute.

    ($\Longleftarrow$) We assume that $I - P$ and $P_\psi$ orthogonally commute. $(I - P) |\psi \rangle = 0 \implies (I - P)|\psi \rangle = |\psi \rangle - P |\psi \rangle = 0 \implies P|\psi \rangle = |\psi \rangle$. Therefore $P$ and $P_\psi$ idempotently commute.
\end{proof}

\section{Event Spaces} \label{app:events}
Usually, when we discuss probabilities, we talk about the sample space, which is the set of the possible outcomes of an event. However, if we want to further represent events that are related to the sample space, one usually has to appeal to an event space, which is usually denoted by the powerset of the sample space. In some of the later sections, we will encounter situations where we want to calculate the probability of the occurrence of a certain event or a collection of events. Let $\Omega$ be a sample space, and let $P(\omega)$ be the associated event space. As shown in \cite{nick-phd}, a probability measure $p: P(\Omega) \rightarrow [0, 1]$ is a function from the event space $P(\Omega)$ to $[0, 1]$ such that
\begin{itemize}
    \item $p(\Omega) = 1$,
    \item For any two events $e_1, e_2 \in P(\Omega)$: $p(e_1 \cup e_2)$ = $p(e_1) + p(e_2) - p(e_1 \cap e_2)$.
\end{itemize}

In standard quantum theory, we use projectors to represent measurement outcomes. More precisely, we represent measurement by an orthogonal and complete set of $n$ projectors $\Omega = \{P^i\}_{i=1}^{n}$. By orthogonal, we mean that all the projectors are pairwise orthogonal $P^iP^j=0$ for $i \neq j$, and by complete we mean that all the projectors in the set $\Omega$ sum to the identity $\sum_{i = 1}^{n} P_i = 1$ \cite{nick-phd}. 

Now, since $\Omega$ is the sample space for the measurement, it is a good guess that the event space for the measurement would be the powerset of the sample space $\Omega$. However, there is a more convenient way to represent the event space of a measurement that has a one-to-one correspondence with the powerset of the event space. Instead of taking the power set of projectors, consider the projector $Q$ obtained by summing the projectors in a set $S$ in $P(\Omega)$ such that: $Q = \sum_{P \in S}P$. Now, let $\text{E}_{\Omega}$ be the set of projectors obtained using the latter for all the sets $S$ in $P(\Omega)$ . Now, observe that for any sets $S_1$ and $S_2$ in $\text{E}_{\Omega}$, we have: $Q_{S_1} \neq Q_{S_2}$, which means that there is a one-to-one correspondence between the elements of $P(\Omega)$  and $\text{E}_{\Omega}$ \cite{nick-phd}. 

Moreover, we can easily check that there is an equivalence between
\begin{itemize}
    \item The complement of a set $S$ in $P(\Omega)$ and $I - Q_S$ in $\text{E}_{\Omega}$,
    \item Disjunctions $S_1 \cup S_2$ in $P(\Omega)$  and $Q_{S_1} + Q_{S_2} - Q_{S1}Q_{S2}$ in $\text{E}_{\Omega}$,
    \item Conjunctions $S_1 \cap S_2$ in $P(\Omega)$ and $Q_{S_1}Q_{S_2}$ in $\text{E}_{\Omega}$. 
\end{itemize}

This means that on top of the one-to-one correspondence, $\text{E}_{\Omega}$ maintains the logical structure of events.

\section{Operator Algebras and Relevant Properties}
We present the following mathematical background from \cite{nick-phd}.
\begin{definition}
\label{operator-algebra}
An algebra of operators is a set of linear vectors on the vectors of a Hilbert space with a structure such that for two operators $M$ and $N$ in an algebra $\mathcal{X}$,
\begin{itemize}
    \item $M \in \mathcal{X} \implies cM \in \mathcal{X}$,
    \item $M \in \mathcal{X} \implies M^{\dagger} \in \mathcal{X}$,
    \item $M,N \in \mathcal{X} \implies MN \in \mathcal{X}$,
    \item $M,N \in \mathcal{X} \implies M + N \in \mathcal{X}$,
    \item $I \in \mathcal{X}$ and $M \in \mathcal{X} \implies IM = M = MI$.
\end{itemize}
\end{definition}

An example of an algebra would be $Op(\mathcal{H})$, the set of all operators in a Hilbert space $\mathcal{H}$. However, $Op(\mathcal{H})$ is not the only algebra one can find in a Hilbert space, other subsets of $Op(\mathcal{H})$ can construct an algebra. In this paper, when we use the term ``algebras"to mean operator algebras per definition \ref{operator-algebra}.

In this dissertation, we need to distinguish between two different but equivalent types of algebras: Schr\"{o}dinger algebras and Heisenberg algebras. Let $U: A \otimes B \rightarrow C \otimes D$ be a unitary transformation from systems $A$ and $B$ to systems $C$ and $D$. Consider algebras $\mathcal{A} \in Op(\mathcal{H_A})$ and $\mathcal{D} \in Op(\mathcal{H_D})$. As shown in \cite{nick-phd}, there is an equivalence between the Schr\"{o}dinger and Heisenberg algebras such that
\begin{align*}
    M_D \in \mathcal{A} &\iff \tilde{M_A} := M_A \otimes I_B \in \tilde{\mathcal{A}}, \\
    M_D \in \mathcal{D} &\iff \tilde{M_D} := U^{-1} (I_C \otimes M_D) U \in \tilde{\mathcal{D}}.
\end{align*}

We recall this very important result about algebras on Hilbert spaces from quantum information from \cite{quantum-info}. 
\begin{proposition}
    Let $\mathcal{H}$ be a Hilbert space, and let $\mathcal{X} \in Op(\mathcal{H})$ be an algebra. there is some decomposition of $\mathcal{H}$ into a direct sum of tensor products
    \begin{center}
        $\mathcal{H} = \oplus_{i = 1}^{n} \mathcal{H}_{L}^{i} \otimes \mathcal{H}_{R}^{i}$
    \end{center}
    such that 
    \begin{center}
        $X \in \mathcal{X} \iff M = \oplus_{i = 1}^{n} M_{L}^{i} \otimes I_{R}^{i}$.
    \end{center}
\end{proposition}

\begin{definition}
\label{commutant-of-algebra}
(Commutant and Commuting Center)\\
Let $\mathcal{H}$ be a finite-dimensional Hilbert space. Let $\mathcal{X}$ be an algebra such that $\mathcal{X} \in Op(\mathcal{H})$. We call $\mathcal{X'}$, the set of all operators in $\mathcal{X}$ that commute with every operator in $\mathcal{X}$ the commutant of algebra $\mathcal{X}$.
The commuting center of an algebra $Z(\mathcal{X}) = \mathcal{X} \cap \mathcal{X'}$ is the set of the operators in $\mathcal{X}$ that commute with all the operators within $\mathcal{X}$. $Z(\mathcal{X})$ is a commutative algebra. Moreover, $Z(\mathcal{X})$ has the form
\begin{center}
    $M \in \mathcal{X} \iff M = \sum_{i = 1}^{n} c_i \pi_i$ for some $\{c_i\}_{i = 1}^{n}$,
\end{center}
such that $\pi_i$ projects onto a subspace $\mathcal{H}_{L}^{i} \otimes \mathcal{H}_{R}^{i}$. As discussed in \cite{nick-phd}, as a consequence of the above,
\begin{center}
    $M \in Z(\mathcal{X})$ is a projector $\iff M \in \mathcal{E}$,
\end{center}
such that $\mathcal{E}_{\Omega}$ is the event space from the sample space $\Omega = \{\pi_{i}\}_{i = 1}^{n}$ (see \Cref{app:events} for more on sample spaces and events spaces). In other words: $\mathcal{E}_{\Omega} = \text{Proj} \cap Z(\mathcal{X})$.
\end{definition}
\section{The Theory of Causal Balance} \label{app:causalbalance}

Excluding the following definition of unitary circuits, the reader can choose to read this chapter in two different ways: either by starting with the section on unitary circuits, and then coming back to the beginning of the chapter for more detail, or by reading it in the intended order by building up from the fundamental idea of the theory. This chapter has been written with both approaches in mind. 

\begin{definition}
    \label{unitary-circuit}
    A unitary circuit is a circuit made of wires representing systems, and boxes representing unitary transformations. The following is an example of a unitary circuit from $A \otimes B$ to $C \otimes D$:\\
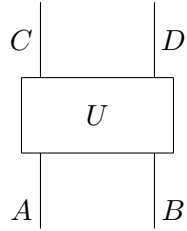
    \begin{center}
        \begin{tikzpicture}
	\begin{pgfonlayer}{nodelayer}
		\node [style=none] (0) at (-14, 4) {};
		\node [style=none] (1) at (-12, 4) {};
		\node [style=none] (2) at (-12, 3) {};
		\node [style=none] (3) at (-14, 3) {};
		\node [style=none] (4) at (-13.75, 3) {};
		\node [style=none] (5) at (-12.25, 3) {};
		\node [style=none] (6) at (-13.75, 4) {};
		\node [style=none] (7) at (-12.25, 4) {};
		\node [style=none] (8) at (-13.75, 5) {};
		\node [style=none] (9) at (-12.25, 5) {};
		\node [style=none] (10) at (-13.75, 2) {};
		\node [style=none] (11) at (-12.25, 2) {};
		\node [style=none] (12) at (-13, 3.5) {$U$};
		\node [style=none] (13) at (-14, 2.25) {$A$};
		\node [style=none] (14) at (-12, 2.25) {$B$};
		\node [style=none] (15) at (-14, 4.5) {$C$};
		\node [style=none] (16) at (-12, 4.5) {$D$};
	\end{pgfonlayer}
	\begin{pgfonlayer}{edgelayer}
		\draw (0.center) to (6.center);
		\draw (6.center) to (7.center);
		\draw (7.center) to (1.center);
		\draw (1.center) to (2.center);
		\draw (2.center) to (5.center);
		\draw (5.center) to (4.center);
		\draw (4.center) to (3.center);
		\draw (3.center) to (0.center);
		\draw (6.center) to (8.center);
		\draw (4.center) to (10.center);
		\draw (5.center) to (11.center);
		\draw (9.center) to (7.center);
	\end{pgfonlayer}
\end{tikzpicture}

    \end{center}
\end{definition}

Ormord and Barrett introduced the theory of causal balance in early 2024. It is considered a conceptual shift from previous theories since causation is no longer merely seen as this connection between events, but a fundamental framework out of which events emerge. 

To present the theory of causal balance in a single chapter, we start by defining interference influences, then discuss operator algebras and how they influence each other. Next, we define causal structure as a directed graph of operator algebras. We will no longer be talking about events emerging out of causation, but rather the emergence of events on a given algebra relative to the set of algebras that include said algebra. One is thus inclined to inquire into the mathematical nature of these events, and the response mirrors that found in the consistent histories formalism: projector decompositions. We will see that the causal structure of this theory restricts the type of influences between projector decompositions, which can also be expressed in how these projector decompositions commute. 

In this theory, events don't emerge from causation in the usual way. Instead, events emerge on a given algebra relative to other algebras that contain it. This raises the question: what are these events mathematically? The answer, following the consistent histories approach, is projector decompositions.
The causal structure limits how projector decompositions can influence each other, which we can see in how they commute. 

Before concluding this chapter, we will demonstrate how the causal structure uniquely selects a set of projector decompositions, which will precisely correspond to a consistent set of histories. Given that the theory of causal balance is inherently stochastic, we will proceed to elucidate how probabilities are defined within this framework. The chapter will conclude with stating the proposition from \cite{nick-paper} that the theory of causal balance can reproduce any phenomenon from standard quantum theory.

\subsection{Interference Influences}
We promised that we will show how the theory of causal balance allows us to obtain a unique consistent set of histories. For now, it is sufficient to convince ourselves that we want a unique set of histories that appeals to a certain causal structure. To understand this structure, we need to understand one of the building blocks of the theory: interference influences.

In the quantum context, we say that system $A$ influences system $D$ in a unitary channel describing the dynamics if $D$ non-trivially depends on $A$, which is illustrated more formally in the following definition \cite{nick-paper}.

\begin{definition} (Quantum Causal Influence)
Let $U: A \otimes B \rightarrow C \otimes D$ be a unitary channel between the systems $A \otimes B$ and $C \otimes D$. Having no quantum causal influence from $A$ to $D$ is equivalent to the following diagram:
\begin{center}
    \begin{tikzpicture}
	\begin{pgfonlayer}{nodelayer}
		\node [style=none] (0) at (-16, 6) {};
		\node [style=none] (1) at (-16, 5) {};
		\node [style=none] (2) at (-14, 5) {};
		\node [style=none] (3) at (-14, 6) {};
		\node [style=none] (4) at (-15.75, 7) {};
		\node [style=none] (5) at (-14.25, 7) {};
		\node [style=none] (6) at (-15.75, 4) {};
		\node [style=none] (7) at (-14.25, 4) {};
		\node [style=none] (8) at (-15.75, 6) {};
		\node [style=none] (9) at (-14.25, 6) {};
		\node [style=none] (10) at (-15.75, 5) {};
		\node [style=none] (11) at (-14.25, 5) {};
		\node [style=none] (12) at (-15, 5.5) {$U$};
		\node [style=none] (13) at (-16, 4.25) {$A$};
		\node [style=none] (14) at (-14, 4.25) {$B$};
		\node [style=none] (15) at (-16, 6.5) {$C$};
		\node [style=none] (16) at (-14, 6.5) {$D$};
		\node [style=none] (17) at (-13, 5.5) {$=$};
		\node [style=none] (18) at (-12, 5) {};
		\node [style=none] (19) at (-12, 4) {};
		\node [style=none] (20) at (-11, 5) {};
		\node [style=none] (21) at (-10, 5) {};
		\node [style=none] (22) at (-10, 6) {};
		\node [style=none] (23) at (-11, 6) {};
		\node [style=none] (24) at (-10.5, 4) {};
		\node [style=none] (25) at (-10.5, 5) {};
		\node [style=none] (26) at (-10.5, 6) {};
		\node [style=none] (27) at (-10.5, 7) {};
		\node [style=none] (28) at (-10.5, 5.5) {$V$};
		\node [style=none] (29) at (-10.25, 4.25) {$B$};
		\node [style=none] (30) at (-10.25, 6.75) {$D$};
		\node [style=none] (31) at (-12.25, 5) {};
		\node [style=none] (32) at (-11.75, 5) {};
		\node [style=none] (33) at (-15.5, 7) {};
		\node [style=none] (34) at (-16, 7) {};
		\node [style=none] (35) at (-18.5, 5.5) {$\exists V$};
		\node [style=none] (36) at (-17.25, 5.5) {such that};
	\end{pgfonlayer}
	\begin{pgfonlayer}{edgelayer}
		\draw (0.center) to (3.center);
		\draw (3.center) to (2.center);
		\draw (2.center) to (1.center);
		\draw (1.center) to (0.center);
		\draw (8.center) to (4.center);
		\draw (5.center) to (9.center);
		\draw (6.center) to (10.center);
		\draw (7.center) to (11.center);
		\draw (23.center) to (26.center);
		\draw (26.center) to (22.center);
		\draw (22.center) to (21.center);
		\draw (21.center) to (25.center);
		\draw (25.center) to (20.center);
		\draw (20.center) to (23.center);
		\draw (26.center) to (27.center);
		\draw (25.center) to (24.center);
		\draw (19.center) to (18.center);
		\draw (31.center) to (18.center);
		\draw (18.center) to (32.center);
		\draw (4.center) to (33.center);
		\draw (4.center) to (34.center);
	\end{pgfonlayer}
\end{tikzpicture}

\end{center}
such that \begin{tikzpicture}
	\begin{pgfonlayer}{nodelayer}
		\node [style=none] (0) at (-16.25, 8) {};
		\node [style=none] (1) at (-15.75, 8) {};
		\node [style=none] (2) at (-16, 7.5) {};
		\node [style=none] (3) at (-16, 8) {};
	\end{pgfonlayer}
	\begin{pgfonlayer}{edgelayer}
		\draw (0.center) to (1.center);
		\draw (3.center) to (2.center);
	\end{pgfonlayer}
\end{tikzpicture}
 stands for the trace operation (``discard" for the readers familiar with string diagrams).
\end{definition}
Can we go deeper and be able to describe the influences from $A$ to $D$ in a more fine-grained way? Yes, we can! It was shown in \cite{nick-paper} that influences between systems are fully determined by influences between subsystems. In fact, as shown in \cite{nick-phd}, interference influences between systems are equivalent to influences between projector decompositions more formally presented in the following definition. 

\begin{definition}
\label{interference-influence}
Let $U: A \otimes B \rightarrow C \otimes D$ be a unitary channel. We say that there is no interference influence from the projector decomposition on the Hilbert space $\mathcal{H}_A$, $\{P_A^{i}\}$ to the projector decomposition on $\mathcal{H}_D$, $\{P_D^{j}\}$ if and only if
\begin{center}
    $[\{P_A^{i}\}, \{P_D^{j}\}] = 0$ for any $i$ and $j$.
\end{center}
\end{definition}

The reason we chose to also define interference influence in terms of the commutation of projector decompositions is that we will later see that it is related to the paradoxes via the commutation link. 

We can begin to see some of the similarities with the consistent sets of histories formalism. Indeed, in the theory of causal balance, events are seen as a unique selection of a projector $P \in \mathbf{D}$ such that $\mathbf{D}$ is a projector decomposition. This projector decomposition is selected in a way that maintains a certain causal balance between systems. In other words, we would only want a specific type of influences between the systems. However, we still have not defined how this causal balance is maintained. What is this magical causal structure that we want to have for which there is only one unique assignment of projector decompositions? The following theorem from \cite{nick-paper} defines how we can obtain the preferred projector decompositions by the theory.

\begin{theorem}
    Let $U: A \otimes B \rightarrow C \otimes D$ be a unitary channel, and let $\{P_A^{i}\}$ be a projective decomposition on $A$. For $\mathcal{A}$ and $\mathcal{D}$ denoting the operator algebras of the form $M_A \otimes I_B$ and $U^{\dagger} (I_C \otimes M_D)$,
    \begin{center}
        $\{P_A^{i}\}$ is preferred by $D \iff \text{span}(\{P_A^{i}\}) \otimes I_B = \mathcal{Z}(\mathcal{A} \cap \mathcal{D'}).$ 
    \end{center}
\end{theorem}

\subsection{Causal Structure}
Since we define causal structure on algebras, we provide the necessary background in \Cref{app:causalbalance}. We start with the following definition from \cite{nick-phd}.
\begin{definition}
\label{causal structure}
A causal structure is a directed graph $\mathfrak{C}$ over a finite set of algebras on a finite-dimensional Hilbert space $\mathcal{H}$ such that for two algebras $\tilde{\mathcal{X}}$ and $\tilde{\mathcal{Y}}$ in $\mathfrak{C}$,
\begin{center}
    $\tilde{\mathcal{X}}$ does not influence $\tilde{\mathcal{Y}}$ and $\tilde{\mathcal{Y}}$ does not influence $\tilde{\mathcal{X}}$ $\iff \tilde{\mathcal{X}} \subseteq \tilde{\mathcal{Y}}$.
\end{center}
\end{definition}

When we talk about a causal structure, we also talk about a bubble. Bubbles are also relevant since in most cases, we want to maintain the causal balance relative to a certain bubble of systems.
\begin{definition}
\label{bubble}
A bubble $\mathfrak{B}$ is a subset of systems in $\mathfrak{C}$. In a unitary circuit, a bubble is a subset of wires.
\end{definition}

The reader might be wondering if there is any structure to these graphs that maintains the causal balance between algebras, and the answer is positive. Consider an algebra $\mathcal{X}$ relative to a bubble $\mathfrak{B}$. We define the following two event spaces (see Appendix 2 for mathematical background on event spaces). 
\begin{itemize}
    \item \textbf{Future balanced event space}: $\mathcal{E}^{\uparrow}_{\tilde{\mathcal{X}} \mathfrak{B}} := \text{Proj} \cap Z(\tilde{\mathcal{X}} \cap \mathfrak{C}_{\mathfrak{B}}^{\uparrow}(\tilde{\mathcal{X}})')$,
    \item \textbf{Past-balanced event space}: $\mathcal{E}^{\downarrow}_{\tilde{\mathcal{X}} \mathfrak{B}} := \text{Proj} \cap Z(\tilde{\mathcal{X}} \cap \mathfrak{C}_{\mathfrak{B}}^{\downarrow}(\tilde{\mathcal{X}})')$,
\end{itemize}
such that $\mathfrak{C}_{\mathfrak{B}}^{\uparrow}(\tilde{\mathcal{X}})$ is the resulting algebra from combining all the algebras in the bubble $\mathfrak{B}$ that $\tilde{\mathcal{X}}$ influences. Similarly, $\mathfrak{C}_{\mathfrak{B}}^{\downarrow}(\tilde{\mathcal{X}})$ is the resulting algebra from combining all the algebras in the bubble $\mathfrak{B}$ that are influenced by $\tilde{\mathcal{X}}$.

If the symbols in the above expressions are confusing, we recommend looking at Appendix 2 on algebras and some of their properties. 

These expressions present exactly what we were looking for. Events spaces (which are projectors) that satisfy theorem  and hence maintain causal balance. Here we come to see how we took advantage of the time-symmetry mentioned earlier when it comes to having past- and future-balanced events. 

\subsection{From Algebras to Unitary Circuits}
Consider the following unitary circuit:
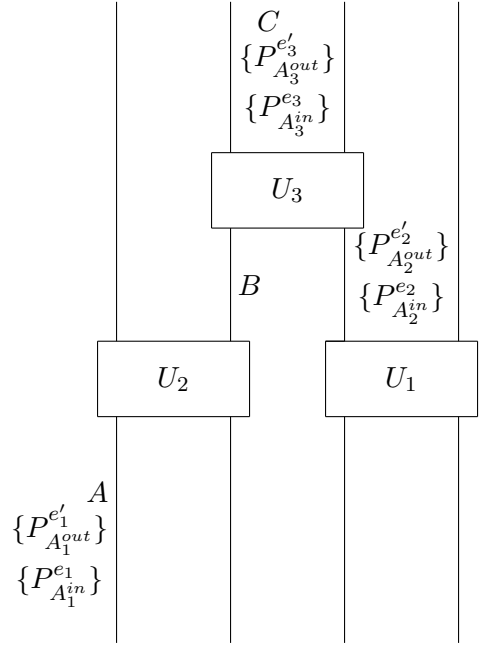
\begin{center}
    
\end{center}

Let's consider the set of system $\{A, B, C\}$. We call a set of systems in a circuit a bubble, and we will name this bubble composed of $A$, $B$, and $C$, $\mathfrak{B}_1$. Now, if we were to think about the projector decompositions that we can have at each system from $\mathfrak{B}_1$ in the context of the consistent histories formalism, one can have too many sets of histories corresponding to $\mathfrak{B}_1$ that are consistent. However, we want to pick one unique set of histories that appeals to a causal structure where all the projectors are causally balanced with respect to their future and their past. This is the reason why we decorated the unitary circuit with two projectors at each system. The $P_{X_{in}}^{i}$ are projectors that are causally balanced with respect to the system $X$'s interaction with its past within $\mathfrak{B}_1$, and $P_{X_{out}}^{i}$ are the causally balanced projectors with respect to the system $X$'s interaction with the future within $\mathfrak{B}_1$.

The definition of a causal structure in terms of algebras and influences between them might not be very intuitive and circuits can make understanding the theory better. We can take advantage of the fact that any unitary circuit defines a model in the theory: when we use unitary circuits, we automatically create a model for our theory as the causal structure is implicitly implied by the position of the system relevant to each other: systems at the top are influenced by systems at the bottom. However, it is indispensable to note that the theory of causal balance is defined independently from spacetime and so it is not evident to deduce that it means that systems at the top necessarily happen before systems at the top. It is more correct for us to think about spacetime emerging from the causal structure just like events \cite{nick-phd}. 

Moreover, it is crucial to remember that not all causal structures can be represented as unitary circuits. Still, unitary circuits do represent some models, and we will focus on those in this section for clarity purposes. 

If the bubble under study is composed of $n$ systems, then, $2n$ events take place represented by $n$ ongoing projections, and $n$ outgoing projections. In other words, each wire $k \in \{1, ..., n\}$ in the bubble has a pair $(\{P_{in}^{k}\}, \{P_{out}^{k}\})$ associated with it. The projector decompositions get selected using a preference algorithm that respects the causal structure. It has been proven that in the context of this interpretation, the only allowed interference influences are the ones from decompositions of the past to decompositions of the future, which is put more formally in theorem \ref{allowed-interferences} in the following section.

Eureka! We finally found out how we get the unique set of projectors associated with a bubble that represents events that \emph{might happen}. We would like to emphasize the latter since the projectors obtained are not events that will happen, these projectors decompositions only correspond to events that are possible to happen, not to events that are about to happen. As shown in \cite{nick-paper}, for a bubble $\mathfrak{B}$, the probability for the unique set that was selected to be realized for a circuit $\mathfrak{C}$ containing $\mathfrak{B}$ is
\begin{align*}
    p(e_1, e_1', ..., e_n, e_n') = \frac{1}{d} Tr\left({\tilde{P}_{A_{1}^{in}}^{e_1}}{\tilde{P}_{A_{1}^{out}}^{e_1'}}...{\tilde{P}_{A_{n}^{in}}^{e_n}}{\tilde{P}_{A_{n}^{out}}^{e_n}}\right).
\end{align*}

It has also been shown in \cite{nick-paper} that this probability rule recovers the exact predictions of quantum theory. 

To conclude our review of the theory of causal balance, we answer the anticipated question about reproducing standard quantum theory. In \cite{nick-paper}, it was proven that the theory of causal balance is able to recreate any phenomenon in standard quantum theory.

\end{document}